\documentclass[11pt,letterpaper]{article}

\usepackage[ruled]{algorithm2e}

\SetAlFnt{\small}
\SetAlCapFnt{\small}
\SetAlCapNameFnt{\small}
\SetAlCapHSkip{0pt}
\IncMargin{-\parindent}

\usepackage{amsmath,amsfonts,graphpap,amscd,mathrsfs,graphicx,lscape}
\usepackage{epsfig,amssymb,amstext,xspace}
\usepackage{theorem}

\usepackage{color}              
\usepackage{epsfig}

\newcommand{\E}{\mathbb{E}}

\newtheorem{theorem}{Theorem}
\newtheorem{lemma}[theorem]{Lemma}

\newtheorem{corollary}[theorem]{Corollary}
\newtheorem{defn}[theorem]{Definition}

{\theorembodyfont{\rmfamily} }
{\theorembodyfont{\rmfamily} }

\def\squareforqed{\hbox{\rule{2.5mm}{2.5mm}}}
\def\blackslug{\rule{2.5mm}{2.5mm}}
\def\qed{\hfill\blackslug}

\def\QED{\ifmmode\squareforqed 
  \else{\nobreak\hfil   
    \penalty50                 
    \hskip1em                  
    \null                      
    \nobreak                   
    \hfil                      
    \squareforqed              
    \parfillskip=0pt           
    \finalhyphendemerits=0     
    \endgraf}                  
  \fi}

\def\blksquare{\rule{2mm}{2mm}}
\def\qedsymbol{\blksquare}
\newcommand{\bg}[1]{\medskip\noindent{\bf #1}}
\newcommand{\ed}{{\hfill\qedsymbol}\medskip}
\newenvironment{proof}{\bg{Proof. }}{\ed}

\newenvironment{proofof}[1]{\bg{Proof of #1: }}{\ed}

\newcommand{\comment}[1]{}

\newcommand{\R}{\ensuremath{\mathbb{R}}}

\newcommand{\C}{\ensuremath{\mathcal{C}}}

\newcommand{\G}{\ensuremath{\mathcal G}}

\newcommand{\opt}{\hbox{OPT}}

\begin{document}

\title{Strong Price of Anarchy and Coalitional Dynamics}

\author{
Yoram Bachrach\thanks{Microsoft Research Cambridge, UK. {\tt yobach@microsoft.com} }
\and
Vasilis Syrgkanis\thanks{Cornell University.
{\tt vasilis@cs.cornell.edu}. Work performed in part while an intern with Microsoft Research. Supported in part by ONR grant N00014-98-1-0589 and NSF grants CCF-0729006 and a Simons Graduate Fellowship.}
\and
\'{E}va Tardos\thanks{Cornell University.
{\tt eva@cs.cornell.edu}}
\and
Milan Vojnovi\'c\thanks{Microsoft Research Cambridge, UK.
{\tt milanv@microsoft.com}}
}

\maketitle

\begin{abstract}
We introduce a framework for studying the effect of cooperation on the quality of outcomes in utility games. Our framework is a coalitional analog of the smoothness framework of non-cooperative games. Coalitional smoothness implies bounds on the strong price of anarchy, the loss of quality of coalitionally stable outcomes, as well as bounds on coalitional versions of coarse correlated equilibria and sink equilibria, which we define as out-of-equilibrium myopic behavior as determined by a natural coalitional version of best-response dynamics.

Our coalitional smoothness framework captures existing results bounding the strong price of anarchy of network design games. We show that in any monotone utility-maximization game, if each player's utility is at least his marginal contribution to the welfare, then the strong price of anarchy is at most $2$. This captures a broad class of games, including games with a very high price of anarchy. Additionally, we show that in potential games the strong price of anarchy is close to the price of stability, the quality of the best Nash equilibrium.
\end{abstract}
\section{Introduction}
We introduce a framework for studying the effect of cooperation on the quality of outcomes in games. In the past decade we have developed a good understanding about 
the degradation in social welfare
in games due to selfish play, quantified by the price of anarchy. There are tight bounds known for the price of anarchy in a range of games from routing, through network design, to various scheduling games. Much less is understood about outcomes of games where players may cooperate.

However, in many settings players do cooperate, and in many games cooperation can help improve the outcome. The worst possible Nash equilibrium is a very pessimistic prediction of the outcome in games that are not strictly competitive, and where cooperation may improve the utility for all participants. A key issue in understanding cooperative outcomes is to what extend the players can transfer utility among each-other. Maybe the two dominant notions of cooperative outcomes that have been considered in the literature are  the strong Nash equilibrium of Aumann \cite{Aumann1959} assuming no utility transfer between players, and the transferable game notion of the core (see \cite{Maschler1992} for a survey). 
An outcome is a strong Nash equilibrium if it is stable subject to coalitional deviations, meaning that no group of players can jointly deviate to improve the solution for every member of the coalition. Allowing utility transfers between the players, leads to a more demanding form of equilibrium, a solution is unstable in this sense, if there is a possible joint deviation for a group that improves the total utility of a group, even if this is not improving the utility of every single player.

The strong price of anarchy was introduced by Andelman et al \cite{Andelman2009} and measures the quality degradation of strong Nash equilibria in games. One of the most compelling examples is the cost-sharing games. Anshelevich at al \cite{Anshelevich2004} showed that the price of anarchy in this class of games with $n$ players can be as bad as $n$, but showed a tight $H_n=O(\log n)$ bound on the price of stability, the quality loss in the best Nash equilibria compared to the socially optimal solution. While the worst Nash equilibria seems too pessimistic a prediction for the outcome, the best Nash equilibria is potentially too optimistic: while significant cooperation is needed to identify and reach this solution, the stability concept used is that of Nash equilibria, assuming that only individual players can deviate, and not groups. Epstein et al \cite{Epstein2009} showed an $H_n$ bound on the the strong price of anarchy, matching the price of stability bound.

The strong Nash equilibria (and the strong price of anarchy) is a compelling outcome prediction in games when strong equilibria exist. However, strong Nash equilibria are not guaranteed to exist, and do not exist in even small and simple cost-sharing games \cite{Epstein2009}. Under the transferable utility definition of coalitional stability, a stable outcome is even less likely to exist. Such stable outcome is automatically socially optimal (or else the group of all players can deviate), requiring at least the price of stability to be 1.  In fact, already Nash equilibria as a prediction of a game outcome has issues. While Nash equilibria are guaranteed to exist, they may not be unique, and natural ``game play'' tends not to converge to Nash equilibria; rather, repeated best response style deviations tend to lead to cycling between outcomes. Hence, it is important to understand the efficiency of games without reaching stable outcomes. We need to find approaches to quantifying the efficiency of coalitionally stable equilibrium outcomes in a manner that they would directly extend with very small degradation even to out-of-equilibrium cooperative dynamic solution concepts.

Roughgarden \cite{Roughgarden2009} introduced a framework, called smooth games, encompassing most price of anarchy bounds, and showing that bounds proved by his smoothness framework automatically extend also to coarse correlated equilibria, which are outcomes of no-regret learning by each player \cite{Blum}.  Extending the price of anarchy results to no-regret outcomes is appealing as it is a natural model of player behavior, and no-regret can be achieved via simple strategies.

The goal of our paper is to initiate a similar study of outcomes of dynamic cooperative play. We
propose a smoothness framework that captures efficiency in most well-established 
cooperative equilibrium solution concepts such as the strong Nash equilibrium and randomized versions of it, and the efficiency guarantees that it implies directly extend with small loss to a form of out-of-equilibrium cooperative dynamics. The solution concepts of Nash equilibria and as well as the learning outcome of coarse correlated equilibrium, is based on the assumption that every player acts independently, in a solely self-interested fashion. We study outcomes of a form of cooperative play, and our goal is to offer conclusions about the efficiency loss of cooperative play that hold even without reaching an equilibrium, including games that don't have strong Nash equilibria.

\paragraph{Our Results.} We propose a general framework for quantifying the quality of strong Nash equilibria by introducing the notion of coalitional smoothness. We show how coalitional smoothness captures existing results on network design games, we give new results on the strong price of anarchy in utility games, and show that coalitional smoothness in such games implies high social welfare at coalitional sink equilibria, which we define as the out-of-equilibrium myopic behavior as defined by a natural coalitional version of best-response dynamics.

\begin{itemize}
\item
 We define the notion of a {\em $(\lambda,\mu)$-coalitionally smooth games} and show that the strong price of anarchy of a $(\lambda,\mu)$-coalitionally smooth game is bounded by $\lambda/(1+\mu)$ in utility games and  $\lambda/(1-\mu)$ in cost minimization games.
\item We show that the cost-charing games of \cite{Epstein2009} as well as network contribution games \cite{Anshelevich2010} studied in the literature are coalitionally smooth.
\item We show that in any monotone utility-maximization game, if each player's utility is at least his marginal contribution to the welfare then the strong price of anarchy is at most $2$, while the price of anarchy in this class of games can be as high as $n$. This result complements the results of Vetta and Goemans et al. \cite{Vetta2002,Goemans2005} who studied the price of anarchy of utility-maximization games that have submodular social welfare function.
\item In potential games, such as the cost-sharing game of \cite{Epstein2009}, the potential minimizer is a Nash equilibrium of high quality. This equilibrium is typically used to bound the price of stability by showing that the social welfare function is similar to the potential function, namely $\lambda \cdot SW(s) \leq  \Phi(s)\leq \mu \cdot SW(s)$, implying a bound of $\lambda/\mu$ on the price of anarchy. We show that in utility games this condition also implies that the game is  $(\lambda,\mu)$-coalitionally smooth implying a  $\lambda/(1+\mu)$ bound on the strong price of anarchy, and give conditions for a similar bound in cost-minimization potential games extending the work of \cite{Epstein2009}.
\item Strong price of anarchy bounds via coalitional smoothness also extend to the notions of strong correlated equilibria of Moreno and Wooders \cite{Moreno1996} and strong coarse correlated equilibria of Rozenfeld et al. \cite{Rozenfeld06strongand}, which correspond to randomized outcomes where no group of players $C$ has a joint distribution of strategies $\tilde{D}_C$ that each member of the group has regret for. Though there exist games with no strong Nash, that admit such randomized strong equilibria, unfortunately, there is no simple game play that guarantees this coalitional no-regret property, and in fact, strong coarse correlated equilibrium may not exist in some games.
\item We define a natural {\em coalitional best response dynamic} and the corresponding {\em coalitional sink equilibria}, the analog of the notion of sink equilibria introduced by Goemans et al. \cite{Goemans2005} for coalitional dynamics. While sink equilibria correspond to steady state behavior of the Markov chain defined by iteratively doing random unilateral best respond dynamics, coalitional sink equilibria are the steady state under our coalitional best response dynamic. We do not explicitly model how players chose to transfer utility to each other. However, our dynamic assumes that when a group cooperates, then they can also transfer utility, and hence will choose to optimize the total utility of all group members. We show that in $(\lambda,\mu)$-coalitionally smooth utility games the social welfare of any coalitional sink equilibrium is at least a $\frac{1}{H_n}\frac{\lambda}{1+\mu}$ fraction of the optimal; extending our analysis of outcomes of coalitional play to games when strong Nash equilibria does not exist.
\end{itemize} 
\subsection*{Related Work}
\label{sec:related}

The study of efficiency of worst-case Nash equilibria via the price of anarchy was initiated by Koutsoupias and Papadimitriou \cite{Koutsoupias1999}, and has
triggered a large body of work. Roughgarden \cite{Roughgarden2009} introduced a canonical way of analyzing the price of anarchy by proposing the notion
of a $(\lambda,\mu)$-smooth game and showing that most efficiency proofs can be cast as showing that the game is smooth. Most importantly, \cite{Roughgarden2009} showed that any efficiency proven via smoothness arguments directly extends to outcomes of no-regret learning behavior. Recently, similar frameworks
have been proposed for games of incomplete information \cite{Roughgarden2012,Syrgkanis2012,Syrgkanis2013} and games with continuous strategy spaces \cite{Roughgarden2010}. However, these frameworks do not take into account coalitional robustness and no canonical way of showing efficiency bounds for coalitional solution concepts existed prior to our work.

The most well-established coalitionally robust solution concept is that of the strong Nash equilibrium introduced by Aumann et al. \cite{Aumann1959}. The study of the efficiency of the worst strong Nash equilibrium (strong price of anarchy) was introduced in~\cite{Andelman2009}, and follow-up research mostly focused on specific cost minimization games such as network design games~\cite{Anshelevich2004,Albers2008,Epstein2009}. Our coalitional smoothness framework captures some of the results in this literature and gives a generic condition under which the strong price of anarchy is bounded.

For utility maximization games Vetta \cite{Vetta2002} defined the class of valid-utility games, which are utility maximization games with a monotone and submodular welfare function and where each player's utility is at least his marginal contribution to the welfare. Vetta \cite{Vetta2002} showed that every Nash equilibrium of a valid utility game achieves at least half of the optimal welfare. Later these games were analyzed from the perspective of best response dynamics by Goemans et al.  \cite{Goemans2005}, who introduced the notion of a sink equilibrium (i.e. steady state distribution of the Markov chain defined
by best-response dynamics) and showed that for a subclass of valid-utility games the half approximation is
achieved after polynomially many rounds, while for the general class, the sink equilibria can have an efficiency that degrades linearly with the number of players.
In this paper, we show that even without the assumption of submodularity every monotone utility maximization game  that satisfies the marginal contribution condition has good strong price of anarchy. Additionally, we define a coalitional version of sink equilibria of Goemans et al. \cite{Goemans2005} and show that for any coalitionally smooth game the efficiency at these out-of-equilibrium dynamics is only a logarithmic in the number of players degradation of the implied
strong price of anarchy bound.

The efficiency of coalitionally robust soluction concepts in the context of utility-maximization games was also studied by Anshelevich et al \cite{Anshelevich2010} for a class of contribution games in networks, where pairwise-stable outcomes where analyzed. Most of our theorems imply social welfare bounds for strong
Nash equilibria of network contribution games, that hold under much more general assumptions than the ones considered in \cite{Anshelevich2010}.

The existence of strong Nash equilibria was examined by both game theorists and computer scientists (see e.g.~\cite{sne_existence,Harks2009,Epstein2009}). Rozenfeld and Tennenholtz~\cite{Rozenfeld06strongand} show that in singleton congestion games with increasing resource value functions there always exists a strong Nash equilibrium, while Holzman et al.~\cite{Holzman1997} show that  for decreasing function the set of pure Nash equilibria, which is non-empty, coincides with the set of strong Nash equilibria.

\section{Coalitional Smoothness}\label{sec:coalitional-smoothness}
In this section we introduce the notion of coalitional smoothness and show that it captures the essence of efficiency guarantees of strong Nash equilibria in several games studied in the past, such as network cost sharing games \cite{Anshelevich2004,Epstein2009} as well as in new classes of games that we give, which generalize the well-studied valid-utility games of Vetta \cite{Vetta2002} to general utility maximization games dropping the assumption of submodularity.

For ease of presentation we will present the definition of coalitional smoothness for utility maximization games rather than cost minimization, but the definitions naturally extend to analogous ones for cost minimization. We will consider a standard normal form game among $n$ players. Each player $i$
has a strategy space $S_i$ and a utility $u_i:S_1\times\ldots \times S_n\rightarrow \R_+$. For a subset of players $C\subseteq [n]$ we will denote
with $S_C=(S_i)_{i\in C}$ the joint strategy space, with $s_C\in S_C$ a joint strategy profile and with $\Delta(S_C)$ the space of distributions over strategy profiles. We are interested in quantifying the efficiency of coalitional solution concepts with respect to the social welfare, which is defined as the sum of
all player utilities: $SW(s)= \sum_{i\in [n]} u_i(s)$. For convenience, we will denote with $\opt$ the maximal social welfare (resp. minimum social cost) achieved among all possible strategy profiles and we will try to upper bound the \emph{price of anarchy}, which is the ratio of the optimal social welfare over the social welfare at any equilibrium in the class of solution concepts that we study (e.g. strong price of anarchy for the case of strong Nash equilibria), or equivalently to lower bound the fraction of the optimal welfare that every equilibrium in the class achieves.

The intuition behind coalitional smoothness is that it requires from the game to admit a good strategy profile such that if enough players coalitionally deviate to this strategy from any state with low social welfare then they achieve a good fraction of the optimal social welfare. In mechanisms this condition can be alternately phrased as requiring that the group of players can achieve a good fraction of the optimal social welfare by not paying much more than the current prices \cite{Syrgkanis2013}. Specifically, it imposes that if we order the players arbitrarily and consider only the coalitional deviations of all the suffixes of this order then the total utilities of the first player in each of the suffixes, after the coalitional deviation of the suffix, is at least a $\lambda$ fraction of the optimal social welfare or else $\mu$ times the current social welfare is at least a $\lambda$ fraction of the optimal.

\begin{defn}[Coalitional Smoothness]
A utility maximization game is $(\lambda,\mu)$-coalitionally smooth if there exists a strategy profile $s^*$ such that for any
strategy profile $s$ and for any permutation $\pi$ of the players:
\begin{equation}
\sum_{i=1}^n u_i(s_{N_{\pi(i)}}^*,s_{-N_{\pi(i)}})\geq \lambda\cdot \opt - \mu\cdot SW(s)
\end{equation}
where $N_{\pi(i)}=\{j\in [n]: \pi(j)\geq \pi(i)\}$ is the set of all players preceding $i$ in the permutation and  $(s_{N_t},s_{-N_t})$
is the strategy profile where all players in $i\in N_t$ play $s_i^*$ and all other players play $s$.
\footnote{In the case of cost minimization games we would require:
$\sum_{i=1}^n c_i(s_{N_{\pi(i)}}^*,s_{-N_{\pi(i)}})\leq \lambda\cdot SC(s^*) + \mu\cdot SC(s)$}
\end{defn}

We now formally define the notion of a strong Nash equilibrium introduced by Aumann \cite{Aumann1959} and show that coalitional smoothness implies high efficiency at every strong Nash equilibrium of a game.
\begin{defn}[Strong Nash Equilibrium] A strategy profile $s$ is a strong Nash equilibrium if for any coalition $C\subseteq [n]$ and for any
coalitional strategy $s_C\in S_C$, there exists a player $i\in C$ such that:
$u_i(s) \geq u_i(s_C,s_{-C})$.
\end{defn}

\begin{theorem}\label{thm:strong-poa}
If a game is $(\lambda,\mu)$-coalitionally smooth for some $\lambda,\mu\geq 0$ then every strong Nash equilibrium has social welfare at least $\frac{\lambda}{1+\mu}$ of the optimal. \footnote{In cost-minimization games $(\lambda,\mu)$-coalitional smoothness for $\lambda\geq 0$ and $\mu\leq 1$ would imply that the social cost at every strong Nash equilibrium is at most $\frac{\lambda}{1-\mu}$ of the minimum cost.}
\end{theorem}
\begin{proof}
Let $s$ be strong Nash equilibrium strategy profile and let $s^*$ be the optimal strategy profile. If all players coalitionally deviate to $s^*$ then, by the definition of a strong Nash equilibrium, there is a player $i$ who is blocking the deviation, i.e. $u_i(s)\geq u_i(s^*)$. Without loss of generality, reorder the players such that this is player $1$. Similarly, if players $\{2,\ldots,n\}$ deviate to playing their strategy in $x^*$ then there exists some player, obviously different than $1$ who is blocking the deviation. Without loss of generality, by reordering we can assume that this player is $2$. Using similar reasoning we can reorder the players such that if players $\{i,\ldots,n\}$ deviate to their strategy in the optimal strategy profile $x^*$ then player $i$ is the one blocking the deviation. That is player $i$'s utility at the strong Nash equilibrium is at least his utility in the deviating strategy profile.
 Thus under this order $\forall k\in N$: $u_i(x)\geq u_i(s_{N_k}^*,s_{-N_k})$. Summing over all players and using the coalitional
smoothness property for the above order we get the result:
\begin{align*}
\textstyle{SW(s) = \sum_{i=1}^N u_i(s) \geq \sum_{i=1}^N u_i(s_{N_i}^*,s_{-N_i})\geq \lambda SW(s^*) - \mu SW(s)}
\end{align*}
\end{proof}

Similar to smoothness, coalitional smoothness also implies efficiency bounds even for randomized coalition-proof solution concepts. Adapting randomized solution
concepts such as correlated equilibria so as to make them robust to coalitional deviations is not as straightforward as in the case of unilateral stability. This is mainly due to information
considerations. One well-studied such concept is that of strong correlated equilibria of Moreno and Wooders \cite{Moreno1996}.

\begin{defn}[Strong Correlated Equilibrium]
A distribution $D\in \Delta(S)$ over strategy profiles is a strong correlated equilibrium, if for any coalition $C\subseteq [n]$ and for any mapping $\tilde{D}_C: S_C\rightarrow \Delta(S_C)$
there exist a player $i\in C$ such that:
\begin{equation}
\E_{s\sim D}[u_i(s)]\geq \E_{s\sim D}\E_{\tilde{s}_C\sim \tilde{D}_C(s_C)}[u_i(\tilde{s}_C,s_{-C})]
\end{equation}
\end{defn}
This notion assumes that coalitions form at the ex-ante stage, before players receive their recommendations of which strategy to play. The deviation is a conditional plan on which
distribution players will deviate too, conditional on the recommendation that they get. Thus implicitly it is assumed that if players commit to a coalition ex-ante, then after receiving
their recommendations on which strategy to play, they share it publicly among the players in the coalition and decide on a joint deviation. A strong correlated equilibrium asks that
for any coalition and for any deviating plan there exists some player in the coalition that doesn't benefit from the deviation. Using the same approach as in Theorem \ref{thm:strong-poa}
and the fact that the deviating strategy designated by the coalitional smoothness property is independent of the same for any strategy profile we get the same efficiency guarantee for strong correlated equilibria too. Unlike Correlated Equilibria, Strong Correlated Equilibria don't always exist even in cost-sharing games, as seen by the example of Epstein at al \cite{Epstein2008} modeling the prisoner's dilemma as a cost-sharing game. However, there are games that admit no Strong Nash Equilibrium but have a Strong Correlated Equilibrium (see e.g. Moreno and Wooders \cite{Moreno1996} for such an example of a three-player matching pennies game).

Similarly, one can define the coalitional equivalent of coarse correlated equilibria, which was analyzed by Rozenfeld et al.~\cite{Rozenfeld06strongand}, under which the coalitional deviation cannot depend on the recommended actions of the players, but is a fixed coalitional strategy.
\begin{defn}[Strong Coarse Correlated Equilibrium]
A distribution $D\in \Delta(S)$ over strategy profiles is a strong coarse correlated equilibrium, if for any coalition $C\subseteq [n]$ and for any distribution $\tilde{D}_C\in \Delta(S_C)$
there exist a player $i\in C$ such that:
\begin{equation}
\E_{s\sim D}[u_i(s)]\geq \E_{s\sim D,\tilde{s}_C\sim \tilde{D}_C}[u_i(\tilde{s}_C,s_{-C})]
\end{equation}
\end{defn}

\begin{theorem}If a utility game is $(\lambda,\mu)$-coalitionally smooth then every strong coarse correlated equilibrium has expected social welfare at least $\frac{\lambda}{1+\mu}$ of the optimal.\footnote{In a cost minimization game is $(\lambda,\mu)$-coalitionally smooth then every strong coarse correlated equilibrium has expected cost at most $\frac{\lambda}{1-\mu}$ of the optimal.}
\end{theorem}
Strong coarse correlated equilibria are not a strict subset of strong Nash equilibria as defined above, since the deviating constraints that strong correlated
equilibria need to satisfy are a superset of those of strong Nash equilibria. However, strong coarse correlated equilibria allow for distributions
over strategy profiles. We could make strong coarse correlated equilibria a subset if we slightly modified the deviation constraints to only pure deviations. Coalitional smoothness would imply bounds for this larger set of equilibria too.

Strong coarse correlated equilibria are related to no-regret repeated game playing: they correspond to limit sequences of game playing
under which no coalition regrets not having formed a coalition $C$ and deviating coalitionally to some fixed strategy $s_C$ (in the sense of every player in the coalition being better off). However, unlike coarse correlated equilibria (e.g. hedge algorithm), there are no known methods that achieve such coalitional stability in the limit. In fact, there cannot be such algorithms for arbitrary games, since there are games that don't admit strong coarse correlated equilibria, such as the prisoner's dilemma. This observation
highlights the study of other types of cyclic dynamics that would lead to good welfare in the limit even without reaching some equilibrium notion and which
would be valid for any game.
In Section
\ref{sec:best-response} we give the first such out-of-equilibrium efficiency guarantees that take into account coalitional
deviations, by introducing a natural version of myopic coalitional best-response dynamics.

\subsection{Non-Submodular Monotone Utility Games}
Consider a utility maximization game in which every player has an $s_i^{out}$ strategy, corresponding to the player not entering the game. Further assume that the game is monotone with respect to participation, i.e. no player can decrease the social welfare by entering the game: $\forall i\in [n], \forall s: SW(s)\geq SW(s_i^{out},s_{-i})$. We show that the coalitional smoothness of such a game is captured exactly by the proportion of the marginal contribution to the social welfare that a player is guaranteed to get as utility.

\begin{theorem}\label{thm:marginal-contribution}
Any monotone utility maximization game is guaranteed to be $\left(\gamma,\gamma\right)$-coalitionally smooth, if 
each player is guaranteed at least a $\gamma$ fraction of his marginal contribution to the social welfare:
\begin{equation}
\textstyle{\forall s: u_i(s)\geq \gamma \left(SW(s)-SW(s_i^{out},s_{-i})\right)}
\end{equation}
\end{theorem}
\begin{proof}
Consider an arbitrary order of the players and let $s^*$ be the strategy profile that maximizes the social welfare.
By the marginal contribution property we have:
\begin{align*}
\textstyle{\sum_{i=1}^n u_i(s_{N_i}^*,s_{-N_i}) \geq \gamma\cdot \sum_{i=1}^n \left(SW\left(s_{N_i}^*,s_{-N_i}\right)-SW\left(s_i^{out},s_{N_{i+1}^*},s_{-N_i}\right)\right)}
\end{align*}
In addition, by the monotonicity assumption the social welfare can only increase when a player enters the game with any strategy:
\begin{align*}
\textstyle{SW(s_i^{out},s_{N_{i+1}}^*,s_{-N_i})\leq SW(s_k,s_{N_{i+1}}^*,s_{-N_i})=SW(s_{N_{i+1}}^*,s_{-N_{i+1}})}
\end{align*}
Combining the above inequalities we get a telescoping sum that yields the desired property:
\begin{align*}
\textstyle{\sum_{i=1}^n u_i(s_{N_i}^*,s_{-N_i})} \geq~& \textstyle{\gamma\cdot \sum_{i=1}^n \left(SW(s_{N_i}^*,s_{-N_i})-SW(s_{N_{i+1}^*},s_{-N_{i+1}})\right)}\\
\geq~& \textstyle{\gamma\cdot SW(s^*) - \gamma\cdot SW(s) = \gamma \cdot \opt - \gamma\cdot  SW(s)}
\end{align*}
Which is exactly the $(\gamma,\gamma)$-coalitional smoothness property we wanted.
\end{proof}

This latter result complements Vetta's results on valid-utility games. A valid-utility game is a monotone utility-maximization game with the
extra constraint that the social welfare is a submodular function (if viewed as a set  function on strategies). As presented by
Roughgarden \cite{Roughgarden2009}, Vetta showed that in any monotone utility-maximization game with a submodular welfare function,
if each player receives a $\gamma$ fraction of their marginal contribution to the welfare, then the game is $(\gamma,\gamma)$-smooth
implying that every Nash equilibrium achieves a $\frac{\gamma}{\gamma+1}$ fraction of the optimal welfare. In the absence of submodularity
there are easy examples where the worst Nash equilibrium doesn't achieve any constant fraction of the optimal welfare, despite
satisfying the marginal contribution condition. However, our result shows that even in the absence of submodularity every such game will
be $(\gamma,\gamma)$-coalitionally smooth, implying that every strong Nash equilibrium will achieve a $\frac{\gamma}{\gamma+1}$ fraction
of the welfare.

It is important to note that the approximate marginal contribution condition and the submodularity condition are very orthogonal ones. For instance, it is
possible that a game satisfies the approximate marginal contribution condition for some constant, but is not submodular or even approximately submodular
under existing definitions of approximate submodularity. In Appendix \ref{sec:welfare-sharing}, we give a class of welfare sharing games,
where our efficiency theorem applies to give constant bounds on the string price of anarchy, whilst the price of anarchy is unbounded
due to the non-submodularity of the social welfare.

\subsection{Network Cost-Sharing Games.}
In this section we analyze the well-studied class of cost sharing games \cite{Anshelevich2004}, using the coalitional smoothness property.
The game is defined by a set of resources $R$ each associated with a cost $c_r$. Each player's strategy space $S_i$ is  a set
of subsets of $R$. The cost of each resource is shared equally among all players that use the resource and a players total cost
is the sum of his cost-shares on the resources that he uses. If we denote with $n_r(s)$ the number of players using resource
$r$ under strategy profile $s$, then:
$c_i(s) = \sum_{r\in s_i} \frac{c_r}{n_r(s)}$.

Epstein et al \cite{Epstein2009} showed that every  strong Nash equilibrium of the above class of games has social cost at most
$H_n$ times the optimal, where $H_n$ is the $n$-th harmonic number. Here we re-interpret that result as showing that network cost-sharing games are $(H_n,0)$-coalitionally smooth. In the next section we show that the analysis of \cite{Epstein2009} can be applied to a more broad class of potential games, showing a strong connection between the price of stability and the string price of anarchy.

\begin{theorem}[Epstein et al.\cite{Epstein2009}]
Cost sharing games are $(H_n,0)$-coalitionally smooth. (App. \ref{sec:cost-sharing-appendix}).
\end{theorem} 

\section{Best Nash Equilibrium vs. Worst Strong Nash Equilibrium}\label{sec:potential-method}

strong Nash equilibria are a subset of Nash equilibria, so in games when strong Nash equilibria exists, the strong price of anarchy cannot be better than the price of stability (the quality of best Nash). In this section we show that in potential games these two notions are surprisingly close.
We show that through the lens of coalitional smoothness there is a strong connection between the analysis of the efficiency of
the worst strong Nash equilibria and the dominant analysis of the best Nash equilibria, for the class of potential games. A game admits a potential function
if there exists a common function $\Phi(s)$ for all players, such that a player's difference in utility from a unilateral deviation is equal to
difference in the potential:
\begin{equation}
u_i(s_i',s_{-i})-u_i(s) = \Phi(s_i',s_{-i})-\Phi(s)
\end{equation}

A large amount of recent work in the algorithmic game theory literature has focused on the analysis of the efficiency of the best Nash equilibrium (price of stability). For the case of potential games the dominant way of analysing the price of stability is the Potential Method: suppose that the potential function is
$(\lambda,\mu)$-close to the social welfare, in the sense that
\begin{equation}
\lambda \cdot SW(s) \leq  \Phi(s)\leq \mu \cdot SW(s),
\end{equation}
for some parameters $\lambda,\mu\geq0$. Then the best Nash equilibrium achieves at least $\frac{\lambda}{\mu}$ of the optimal social welfare. The proof relies on the simple fact that the potential maximizer is always a Nash equilibrium and by the $(\lambda,\mu)$ property it's easy to see that the potential maximizer has social welfare that is the above fraction
of the optimal social welfare.

The following theorems show that for such potential games the price of stability is very close to the strong price of anarchy, i.e. the implied quality of the best Nash equilibrium is close to the quality of the worst strong Nash equilibrium.

\begin{theorem}\label{thm:potential_spoa}
In a utility-maximization potential game with non-negative utilities, if the potential is $(\lambda,\mu)$-close to the social welfare then
the game is $(\lambda,\mu)$-coalitionally smooth, implying that every strong Nash equilibrium achieves at least $\frac{\lambda}{1+\mu}$ of the optimal social welfare.
\end{theorem}
\begin{proof} Consider an arbitrary order of the players and some strategy profile $s$. By the definition of the potential function and the fact that utilities are non-negative, we have
\begin{align*}
u_i(s_{N_i}^*,s_{-N_i}) = \Phi(s_{N_i}^*,s_{-N_i})-\Phi(s_{N_{i+1}}^*,s_{-N_{i+1}})+u_i(s_{N_{i+1}}^*,s_{-N_{i+1}})
\geq  \Phi(s_{N_i}^*,s_{-N_i})-\Phi(s_{N_{i+1}}^*,s_{-N_{i+1}})
\end{align*}
Combining with our assumption on the relation between potential and social welfare we obtain the coalitional smoothness property:
\begin{align*}
\textstyle{ \sum_i u_i(s_{N_i}^*,s_{-N_i})} \geq~& \textstyle{\sum_i \Phi(s_{N_i}^*,s_{-N_i})-\Phi(s_{N_{i+1}}^*,s_{-N_{i+1}})
  =~ \Phi(s^*)-\Phi(s) \geq \lambda \cdot  SW(s^*)-\mu\cdot SW(s)}
\end{align*}
\end{proof}

Observe that the $(\lambda,\mu)$-closeness of the potential function does not imply smoothness of the game according to the standard definition of smoothness \cite{Roughgarden2009} and hence does not imply a price of anarchy bound. It does so only if the potential is a submodular function and by following a similar analysis as in the case of valid utility games as we show in Appendix \ref{sec:submodular-potential}. Such a property for instance, holds in any utility congestion game with decreasing resource utilities. However, the theorem above does not require submodularity of the potential but only requires the weaker notion of coalitional smoothness to hold.

One example application of the above theorem is in the context of network contribution games \cite{Anshelevich2010}. In a network
contribution game each player corresponds to a node in a social network. Each edge corresponds to a "friendship"  between the
connecting nodes or more generally some joint venture. Each player has a budget of effort that he chooses how to distribute among his friendships. Each friendship
$e$ between two players $i$ and $j$, has a  value $v_e(x_i,x_j)$ that corresponds to the value produced as a function of the
efforts put into it by the two players. This value is equally split among the two players. It is easy to see that in such a game the
social welfare is the total value produced, while the potential is equal to half of the social welfare. Thus, by applying Theorem \ref{thm:potential_spoa}
we get that for arbitrary "friendship" value functions $v_e(\cdot,\cdot)$  the game is $\left(\frac{1}{2},\frac{1}{2}\right)$-coalitionally smooth and
hence every strong Nash equilibrium achieves at least $1/3$ of the optimal social welfare. In contrast, observe that in such a game
Nash equilibria can have unbounded inefficiency, and the game is not $(\lambda,\mu)$-smooth under the unilateral notion of smoothness
for no $\lambda,\mu$.
\footnote{Consider an example of a line of four nodes $(A,B,C,D)$. Each player has a budget of $1$. Edges $(A,B)$ and $(C,D)$ have
a constant value of $1$, while edge $(B,C)$ has a huge value $H$ if both players place their whole budget on it and $0$ otherwise. Players
$B,C$ placing their budget on their alternative friendships is a Nash equilibrium, but not a strong Nash equilibrium.}

For settings where a player can only have non-negative externalities on the utilities of other players by entering the game, a much stronger connection can be drawn. More concretely, a utility maximization game has non-negative externalities if for any strategy
profile $s$ and for any pairs of players $i,j$: $u_i(s)\geq u_i(s_j^{out},s_{-j})$.\footnote{Similarly a cost-minimization game has non-negative
externalities if $c_i(s)\leq c_i(s_j^{out},s_{-j})$.} The $s_i^{out}$ strategy is not required to be a valid strategy that the player can actually pick, but rather a hypothetical strategy, requiring the property that the cost of the player in that strategy is $0$, and the cost functions and the potential are extended appropriately such that the potential function property is maintained even in this augmented strategy space and the potential when all players have left the game is $0$: $\Phi(s^{out})=0$. For instance, every congestion game has the above property if we define the $s_i^{out}$ strategy to be the empty set of resources.

\begin{theorem}\label{thm:positive-potential}
A utility-maximization potential game with only positive externalities and such that $\Phi(s)\geq \lambda\cdot SW(s)$ is $(\lambda,0)$-coalitionally smooth. Similarly, a cost-minimization, potential game with only positive externalities and such that $\Phi(s) \leq \lambda \cdot SC(s)$ is $(\lambda,0)$-coalitionally smooth.
\end{theorem}

In the context of cost-minimization, one well-studied example of such a setting is that of network cost-sharing games and
the $log(n)$ strong price of anarchy result of Epstein et al. \cite{Epstein2009} is a special instance of the above theorem.
In the context of utility-maximization games one example is that of network contribution games under the
restriction that friendship value functions $v_e(\cdot,\cdot)$ are increasing in both coordinates. Under this
restriction applying Theorem \ref{thm:positive-potential} we get the improved bound that every strong Nash equilibrium achieves at least $1/2$
of the optimal social welfare. 

\section{Coalitional Best-Response Dynamics}\label{sec:best-response}
In this section we initiate the study of out of equilibrium dynamic behavior in games. We show that if a utility game is $(\lambda,\mu)$-coalitionally smooth then this implies an efficiency guarantee for out of equilibrium dynamic behavior in a certain best-response like dynamic. This is particularly interesting for games that do not admit a strong Nash equilibrium, but where coalitional deviations are bound to occur. Our approach is similar in spirit to the notion of sink equilibria introduced by Goemans et al. \cite{Goemans2005}. sink equilibria correspond to steady state behavior of the Markov chain defined by iteratively doing random unilateral best respond dynamics. However, such a notion does not capture settings where players can communicate and at each step perform coalitional deviations.

We introduce a version of coalitional best-response dynamics, that allows for coalitional deviations at each time step, giving more probability to small coalitions. In our dynamic, at each step a selected group is chosen to to cooperate.
We assume that when a group cooperates, then they can also transfer utility, and hence will choose to optimize the total utility of all group members.
Then we analyze the social welfare of the steady states arising in the long run as we perform coalitional best response dynamics for a long period. Similar to \cite{Goemans2005} we will refer to these steady states as coalitional sink equilibria. Similar to sink equilibria that are a way of studying games whose best response dynamics might not converge to a pure Nash equilibrium or even games that do not admit a pure Nash equilibrium, coalitional sink equilibria are an interesting alternative for analyzing efficiency in games that do not admit a strong Nash equilibrium, which admittedly is even more rare than the pure Nash equilibrium.

Our coalitional best response dynamics are defined as follows: At each iteration a coalition is picked at random by a distribution that favors coalitions of  smaller size. Specifically, first a coalition size $k$ is picked inversely proportional to the size and then a coalition of size $k$ is picked uniformly at random. Subsequently, the picked coalition deviates to the joint strategy profile that maximizes the total utility of the coalition, conditional on the current strategy of every player outside of the coalition (a more formal definition is given in Algorithm \ref{alg:best-response} in the Appendix).

\begin{theorem}\label{thm:best-response}
If a utility maximization game with non-negative utilities is $(\lambda,\mu)$-coalitionally smooth then the expected social welfare at every coalitional Sink Equilibrium is at least $\frac{1}{H_n}\frac{\lambda}{1+\mu}$ of the optimal.
\end{theorem}
\begin{proof}
Let $s^{t}$ be the strategy profile at some step of the best response dynamics and let $s^*$ be the strategy profile designated by the coalitional smoothness property. We examine the expected welfare of the dynamics after one step. Let $\C_k$ be all the possible coalitions of size $k$.
\begin{align*}
\E[SW(s^{t})~|~s^{t-1}=s] =~& \frac{1}{H_n} \sum_{k=1}^{n} \frac{1}{k} \sum_{C\in \C_k}\frac{1}{\binom{n}{k}}SW(s_{C}^t,s_{-C})
\geq ~ \frac{1}{H_n} \sum_{k=1}^{n} \frac{1}{k} \sum_{C\in \C_k}\frac{1}{\binom{n}{k}}\sum_{i\in C}u_i(s_{C}^t,s_{-C})
\end{align*}
Since the deviation $s_{C_t}^t$ maximizes the total utility of the deviating players, it achieves at least as much welfare as $s_{C_t}^*$:
\begin{align*}
\E[SW(s^{t})~|~s^{t-1}=s]\geq& \frac{1}{H_n} \sum_{k=1}^{n} \frac{1}{k} \sum_{C\in \C_k}\frac{1}{\binom{n}{k}}\sum_{i\in C}u_i(s_{C}^*,s_{-C})
= \frac{1}{H_n} \frac{1}{n!}\sum_{k=1}^n \sum_{C\in \C_k}\sum_{i\in C}\frac{(n-k)! \cdot k!}{k}u_i(s_{C}^*,s_{-C})\\
=~& \frac{1}{H_n} \frac{1}{n!}\sum_{k=1}^n \sum_{C\in \C_k}\sum_{i\in C}(n-k)! \cdot (k-1)!\cdot u_i(s_{C}^*,s_{-C})\\
=~& \frac{1}{H_n} \frac{1}{n!}\sum_{i\in [n]}\sum_{k=1}^{n}\sum_{C\in \C_k: i\in C}(n-k)!\cdot (k-1)!\cdot u_i(s_{C}^*,s_{-C})
\end{align*}
Let $\Pi$ be the set of permutations of players. We argue that:
\begin{align*}
\textstyle{\sum_{i\in [n]}\sum_{k=1}^{n}\sum_{C\in \C_k: i\in C}(n-k)!\cdot (k-1)!\cdot u_i(s_{C}^*,s_{-C})=
\sum_{\pi \in \Pi}\sum_{i=1}^n u_i(s_{N_{\pi(i)}}^*,s_{-N_{\pi(i)}})}
\end{align*}
Observe that for any player $i$ and for any set of players $C\in \C_k$, such that $i\in C$, the term $u_i(s_{C}^*,s_{-C}^{t-1})$ appears in the summation on the right hand side, exactly $(n-k)!\cdot (k-1)!$ times. It appears only when player $i$ is placed at the $k$-th last position in the permutation and it appears once for each possible permutation of the $k-1$ players following $i$ and for each possible permutation of the $n-k$ players preceding $i$. The latter is exactly $(n-k)!\cdot k!$. Using the coalitional smoothness property for each of the permutations we get a lower bound on the welfare at time step $t$:
\begin{align*}
\E[SW(s^{t})~|~s^{t-1}=s]\geq~& \frac{1}{H_n} \frac{1}{n!} \sum_{\pi \in \Pi}\sum_{i=1}^n u_i(s_{N_{\pi(i)}}^*,s_{-N_{\pi(i)}})
\geq~ \frac{1}{H_n} \frac{1}{n!} \sum_{\pi \in \Pi} \left( \lambda \cdot \opt - \mu\cdot SW(s)\right)\\
=~& \frac{1}{H_n} \left( \lambda \cdot \opt - \mu\cdot SW(s)\right)
\end{align*}
Let $D$ be a steady state distribution over strategy profiles of the coalitional best response dynamics. By the definition of the steady state we get:
\begin{equation*}
\E_{s\sim D}\E_{s^t}[SW(s^{t})|s^{t-1}=s]=\E_{s^{t-1}}[SW(s^{t-1})]=E_{s\sim D}[SW(s)]
\end{equation*}
Using this property and our lower bound on the social welfare at time step $t$ conditional on any possible current state $s$ we get:
\begin{align*}
\E_{s\sim D}[SW(s)]=~\E_{s\sim D}\E_{s^t}[SW(s^{t})|s^{t-1}=s]\geq
~ \frac{1}{H_n} \left( \lambda \cdot \opt - \mu\cdot \E_{s\sim D}[SW(s)]\right)
\end{align*}
which yields the claimed lower bound on the expected welfare at the steady state.
\end{proof}

The Markov chain defined by the coalitional best-response dynamics might take long time to converge to a steady state. However, our analysis shows a 
stronger statement: at any iteration $T$ if we take the empirical distribution defined by the best-response play up till time $T$, then the expected welfare of this empirical distribution is at least
$\approx \frac{1}{2}\frac{\lambda}{H_n+\mu}$ of the optimal welfare.
\begin{corollary}\label{cor:any-T}
The empirical distribution of play defined by doing random coalitional best responses for $T$ time steps, achieves
expected social welfare at least $\frac{T-1}{2T}\frac{\lambda}{H_n+\mu}$ of the optimal welfare.
\end{corollary}

\bibliographystyle{abbrv}
\bibliography{reward_bib}

\begin{thebibliography}{10}

\bibitem{Albers2008}
S.~Albers.
\newblock On the value of coordination in network design.
\newblock In {\em SODA}, 2008.

\bibitem{Andelman2009}
N.~Andelman, M.~Feldman, and Y.~Mansour.
\newblock {Strong price of anarchy}.
\newblock {\em Games and Economic Behavior}, 65(2):289--317, Mar. 2009.

\bibitem{Anshelevich2004}
E.~Anshelevich, A.~Dasgupta, J.~Kleinberg, E.~Tardos, T.~Wexler, and
  T.~Roughgarden.
\newblock The price of stability for network design with fair cost allocation.
\newblock In {\em Proceedings of the 45th Annual IEEE Symposium on Foundations
  of Computer Science}, FOCS '04, pages 295--304, Washington, DC, USA, 2004.
  IEEE Computer Society.

\bibitem{Anshelevich2010}
E.~Anshelevich and M.~Hoefer.
\newblock {Contribution Games in Networks}.
\newblock {\em Algorithmica}, pages 1--37, 2011.

\bibitem{Aumann1959}
R.~J. Aumann.
\newblock {Acceptable points in general cooperative N-person games}.
\newblock In R.~D. Luce and A.~W. Tucker, editors, {\em Contribution to the
  theory of game IV, Annals of Mathematical Study 40}, pages 287--324.
  University Press, 1959.

\bibitem{Blum}
A.~Blum and Y.~Mansour.
\newblock chapter Learning, Regret Minimization and Equilibria.
\newblock Camb. Univ. Press, '07.

\bibitem{Epstein2009}
A.~Epstein, M.~Feldman, and Y.~Mansour.
\newblock {Strong equilibrium in cost sharing connection games}.
\newblock {\em Games and Economic Behavior}, 67(1):51--68, 2009.

\bibitem{Goemans2005}
M.~Goemans, V.~Mirrokni, and A.~Vetta.
\newblock Sink equilibria and convergence.
\newblock In {\em Proceedings of the 46th Annual IEEE Symposium on Foundations
  of Computer Science}, FOCS '05, pages 142--154, Washington, DC, USA, 2005.
  IEEE Computer Society.

\bibitem{Harks2009}
T.~Harks, M.~Klimm, and R.~Mohring.
\newblock Strong nash equilibria in games with the lexicographical improvement
  property.
\newblock In {\em WINE}, 2009.

\bibitem{Holzman1997}
R.~Holzman and N.~Law-Yone.
\newblock Strong equilibrium in congestion games.
\newblock {\em Games and Economic Behavior}, 21(1–2):85 -- 101, 1997.

\bibitem{Koutsoupias1999}
E.~Koutsoupias and C.~Papadimitriou.
\newblock Worst-case equilibria.
\newblock In {\em STACS}, 1999.

\bibitem{Maschler1992}
M.~Maschler.
\newblock Chapter 18 the bargaining set, kernel, and nucleolus.
\newblock volume~1 of {\em Handbook of Game Theory with Economic Applications},
  pages 591 -- 667. Elsevier, 1992.

\bibitem{Moreno1996}
D.~Moreno and J.~Wooders.
\newblock Coalition-proof equilibrium.
\newblock {\em Games and Economic Behavior}, 17(1):80--112, November 1996.

\bibitem{sne_existence}
R.~Nessah and G.~Tian.
\newblock On the existence of strong nash equilibria.
\newblock Working Papers 2009-ECO-06, IESEG School of Management, 2009.

\bibitem{Roughgarden2009}
T.~Roughgarden.
\newblock {Intrinsic robustness of the price of anarchy}.
\newblock In {\em STOC}, 2009.

\bibitem{Roughgarden2012}
T.~Roughgarden.
\newblock The price of anarchy in games of incomplete information.
\newblock In {\em ACM EC}, 2012.

\bibitem{Roughgarden2010}
T.~Roughgarden and F.~Schoppmann.
\newblock {Local smoothness and the price of anarchy in atomic splittable
  congestion games}.
\newblock In {\em SODA}, 2011.

\bibitem{Rozenfeld06strongand}
O.~Rozenfeld and M.~Tennenholtz.
\newblock Strong and correlated strong equilibria in monotone congestion games.
\newblock In {\em WINE}, 2006.

\bibitem{Syrgkanis2012}
V.~Syrgkanis.
\newblock {Bayesian Games and the Smoothness Framework}.
\newblock {\em ArXiv e-prints}, Mar. 2012.

\bibitem{Syrgkanis2013}
V.~Syrgkanis and E.~Tardos.
\newblock Composable and efficient mechanisms.
\newblock In {\em STOC}, 2013.

\bibitem{Vetta2002}
A.~Vetta.
\newblock {Nash equilibria in competitive societies, with applications to
  facility location, traffic routing and auctions}.
\newblock 2002.

\end{thebibliography}

\newpage
\begin{appendix}

\section{Welfare Sharing Games without Submodularity}\label{sec:welfare-sharing}
We consider here an interesting example of monotone utility games without the submodularity assumption needed for the valid-utility games of Vetta et al. \cite{Vetta2002}.  Our efficiency theorem applies to this class to give constant bounds on the strong price of anarchy, while the price of anarchy is unbounded due to the
non-submodularity of the social welfare.

Our Welfare Sharing Game is defined as follows: The game is defined by a set of
projects $[m]$. Each player $i\in [n]$ participates in a set of project $P_i\subseteq [m]$ and has a budget of effort
$B_i$ that he chooses how to split among his projects.
Denote with $N_j$ the set of players that participate in project $j\in [m]$ and with $x_j=(x_j^i)_{i\in N_j}$ the
vector of efforts placed at project $j$ by its participants. Each project $j$ is associated with a value function
$v_j(x_j)$, that is monotone in every coordinate and we denote with $\partial_i v_j(x_j) = v_j(x_j)-v_j(0^i,x_{j}^{-i})$ the marginal contribution of player
$i$ to the value. The value of each project is split among the participants proportional to the marginal contribution and the utility of
each participant is the sum of his shares:
\begin{equation}
u_i(x) = \sum_{j\in P_i} \frac{\partial_i v(x_j)}{\sum_{k\in N_j}\partial_k v(x_j)}v(x_j)
\end{equation}

We examine the case where players can be categorized in groups $\G$ where each
group has a specific skill. We denote with $\G_j\subseteq \G$ the subset of skills that contribute to a specific project $j$.

We assume that the value function of a project satisfies the decreasing marginal contribution property
skill-wise. That is, for any given effort levels $x_{-g}$ of players outside of a given skill group $g\in \G_j$,
the group restricted value functions $v_j(x_g,x_{-g})$ satisfy the decreasing marginal contribution property (i.e.
the marginal contribution of a player in the group decreases, if we increase the efforts of other players within the group).
For instance, the value could be of the form of a product of skill-specific functions: $v_j(x_j)=\prod_{g\in G_j}v_j^g(x_g)$,
such that each $v_j^g(x_g)$ is a function that satisfies the decreasing marginal property. This way we could
capture settings where a positive effort from each skill is needed to produce any value. If the set of skill groups
is a singleton and hence the values satisfy the decreasing marginal contribution in general, then the game becomes
a valid-utility game according to Vetta \cite{Vetta2002} and therefore, even Nash equilibria have good social welfare.
However, even if $|\mathcal{G}|=2$ the efficiency of Nash equilibria can be unboundedly worse than the optimal. On the
contrary we show that strong Nash equilibria always achieve a $\frac{1}{\max_j |\mathcal{G}_j|+1}$  fraction of the optimal welfare.

\begin{lemma}
The  Welfare Sharing Game is $(\frac{1}{\max_j |\mathcal{G}_j|},\frac{1}{\max_j |\mathcal{G}_j|})$-coalitionally smooth.
\end{lemma}
\begin{proof}
By monotonicity of the value functions we know that the game is a monotone utility-maximization game. Thus
we simply need to show that it satisfies the approximate marginal contribution property with $\gamma=\frac{1}{|\mathcal{G}|}$
and then Theorem \ref{thm:marginal-contribution} will apply to give the result.

To achieve this we simply need to show that $\sum_j \partial_j v_j(x_j)\leq |\mathcal{G}_j|v_j(x_j)$ since that would imply that:
\begin{equation}
u_i(x) \geq \frac{1}{|\mathcal{G}_j|}\sum_{j\in P_i}\partial_i v_j(x_j)
\end{equation}
Observe that
\begin{align*}
 \sum_{j\in M} \partial_j v_j(x_j) =~& \sum_{g\in \mathcal{G}_j}\sum_{k\in g}\partial_k v_j(x_j)
\end{align*}
Within a skill group the value function satisfies the diminishing marginal contribution
property when effort levels of players outside the skill group are fixed.
For a specific skill group $g$ that contributes to project $j$ consider an arbitrary
ordering of the player $\{1,\ldots,n_g\}$. By the diminishing marginal contribution
property: $\forall k \in g: \partial_k v_j(x_j)\leq
\partial_k v_j(0,\ldots,0,x_{k,j},\ldots,x_{n_g,j},x_j^{-g})$. Then summing over
all players in the group we obtain: $\sum_{k\in g} \partial_k v_j(x_j)\leq v_j(x_j)-v_j(0,\ldots,0,x_j^{-g})\leq v_j(x_j)$.
This then directly implies: $\sum_{g\in \G_j}\sum_{k\in g}\partial_k v_j(x_j)\leq |\G_j|v_j(x_j)$
which completes the proof.
\end{proof}

\section{Network Cost-Sharing Games and Coalitional Smoothness}\label{sec:cost-sharing-appendix}

\begin{theorem}[Epstein et al.\cite{Epstein2009}]
Cost Sharing Games are $(H_n,0)$-coalitionally smooth, where $H_n$ is the $n$-th harmonic number.
\end{theorem}
\begin{proof}
To prove our coalitional smoothness property we need to show that for any ordering of the players:
\begin{equation}
\sum_{i=1}^n c_i(s_{N_i}^*,s_{-N_i})\leq H_n\cdot SC(s^*)
\end{equation}
It is easy to observe that the potential $\Phi(s)$ remains a potential, even if we augment the strategy of each player by allowing them to drop out of the game and use no resource, incurring a cost of $0$. Similar to utility maximization games, we will denote such a strategy as $s_i^{out}$, though we will not assume that this strategy is a strategy that is available to the players, but rather will only use it for the analysis. Using the potential property in this augmented strategy space and since cost shares are increasing as we remove players we have:
\begin{equation}
c_i(s_{N_i}^*,s_{-N_i})\leq c_i(s_{N_i}^*,s_{-N_i}^{out}) = \Phi(s_{N_i}^*,s_{-N_i}^{out})- \Phi(s_{N_{i+1}}^*,s_{-N_{i+1}}^{out})
\end{equation}
Using the above inequality we get a telescoping sum that yields the coalitional smoothness property:
\begin{align*}
\sum_{i=1}^n c_i(s_{N_i}^*,s_{-N_i}) \leq~& \sum_{i=1}^n \left(\Phi(s_{N_i}^*,s_{-N_i}^{out})- \Phi(s_{N_{i+1}}^*,s_{-N_{i+1}}^{out})\right)\\
 =~& \Phi(s^*) - \Phi(s^{out}) = \Phi(s^*) \leq H_n\cdot SC(s^*)
\end{align*}
\end{proof}

\section{ Games with Submodular Potential}\label{sec:submodular-potential}
In this section we present a side-result that potential games with a potential that is $(\lambda,\mu)$-close to the
social welfare are $(\lambda,\mu)$-smooth and therefore have a price of anarchy of $\frac{\lambda}{1+\mu}$. Additionally,
they behave similar to basic utility games and for instance, the above price of anarchy bound is approximately reached even after a polynomial number of random best-response steps.

In this section we will overload notation and talk about social welfare and potential as set functions defined on sets of strategies. Thus we will write $\Phi(\cup_i s_i)$ and $SW(\cup_i s_i)$ the potential and the social welfare when each player chooses a strategy $s_i$. We will assume that the potential and the social welfare are also defined for multisets and with $\cup$ we will denote the
multiset union (i.e. every element is added many times). Under such notation a the potential is
submodular if for any two multisets of strategies $s\subseteq t$ (where inclusion is the extended inclusion for multisets, i.e. each element should appear at least as many times in $t$ as in $s$) and for any strategy $s_i$ of some player $i$:
\begin{equation}
\Phi(s_i \cup s)-\Phi(s)\geq \Phi(s_i\cup t)-\Phi(t)
\end{equation}
We will also assume that each player has an empty strategy and that his utility from this strategy is $0$.

\begin{theorem}\label{thm:submodular-potential}
In any utility-maximization potential game, if the potential is $(\lambda,\mu)$-close to the social welfare
and is monotone submodular, then the game is $(\lambda,\mu)$-smooth and therefore every Nash Equilibrium achieves
at least $\frac{\lambda}{1+\mu}$ of the optimal social welfare.
\end{theorem}
\begin{proof}
Consider any strategy profile $s$ and let $s^*$ be the optimal strategy profile. Let $s_{-i}$ denote
the mutliset  $\cup_{j\neq i}s_j$ , $s_{N_i} = \cup_{j\geq i} s_j$ and $s_{-N_i}=\cup_{j<i}s_j$. By the definition of the potential and the submodularity we have:
\begin{align*}
\sum_{i} u_i(s_i^*,s_{-i}) =~& \sum_{i} \Phi(s_i^*\cup s_{-i})-\Phi(s_{-i})\\
 \geq & \sum_i \Phi(s_{i}^*\cup s_{-i} \cup s_i \cup s_{N_{i+1}}^*)- \Phi(s_{-i} \cup s_i \cup s_{N_{i+1}}^*)\\
 \geq & \sum_i \Phi(s_{N_{i}}^* \cup s)- \Phi( s_{N_{i+1}}^* \cup S)\\
 = & \Phi(s^*\cup s)-\Phi(s) \geq \Phi(s^*)-\Phi(s) \\
 \geq & \lambda SW(s^*)-\mu SW(s)
\end{align*}
\end{proof}

Now we also note that the above bound can be achieved by a polynomial number of best response dynamics. In a random best response dynamic,
at each time step a player is chosen at random and he performs a best response to the current strategies of the rest of the players.
\begin{theorem}\label{thm:submodular-dynamics}
Consider a sequence of random-player best response dynamics. Then after $n\cdot\frac{\lambda}{\lambda+1}\log\left(\frac{\lambda}{(\lambda+1)\epsilon}\right)$ steps
the social welfare is at least $\frac{\lambda}{\mu}\left(\frac{\lambda}{\lambda+1}-\epsilon\right)$ of the optimal.
\end{theorem}
\begin{proof}
Let $\tilde{s}$ be the strategy profile that maximizes the potential and $s^*$ the strategy profile that maximizes the social welfare. At each time step a player is chosen at random and
he best responds given the current strategy profile. We denote with $s^{t}$ the strategy profile at time step $t$. After a step of random best response moves the expected increase in the potential is equal to the expected change of the players' utilities:
\begin{align*}
\E\left[\Phi(s^{t+1})-\Phi(s^{t})\right] =~& \frac{1}{n} \sum_i \left(u_i(s_i^{t+1},s_{-i}^{t})-u_i(s^t)\right)=
\frac{1}{n} \sum_i u_i(s_i^{t+1},s_{-i}^t) - \frac{1}{n} SW(s) \\
\geq~& \frac{1}{n} \sum_i u_i(\tilde{s},s_{-i}^t) - \frac{1}{n} SW(s)
\end{align*}
By using similar reasoning as in Theorem \ref{thm:submodular-potential} we get that $\sum_i u_i(\tilde{s},s_{-i}^t)\geq \Phi(\tilde{s})-\Phi(s^t)$. Thus we get:
\begin{align*}
\E\left[\Phi(s^{t+1})-\Phi(s^{t})\right] \geq \frac{1}{n}\left(\Phi(\tilde{s})-\Phi(s^t)-SW(s)\right)\geq \frac{1}{n}\left(\Phi(\tilde{s})-\left(1+\frac{1}{\lambda}\right)\Phi(s^t)\right)
\end{align*}
Thus the expression $\Phi(\tilde{s})-\left(1+\frac{1}{\lambda}\right)\E[\Phi(s^t)]$, decreases in expectation by at least
$\left(1+\frac{1}{\lambda}\right)\frac{1}{n}(\Phi(\tilde{s})-\left(1+\frac{1}{\lambda}\right)\E[\Phi(s^t)])$, after every time step. Thus if $s^0$ is the initial strategy profile, after $t$ time steps we have:
\begin{equation}
\Phi(\tilde{s})-\left(1+\frac{1}{\lambda}\right)\E[\Phi(s^t)]\leq \left(1-\left(1+\frac{1}{\lambda}\right)\frac{1}{n}\right)^{t}\left(\Phi(\tilde{s}) - \Phi(s^0)\right)\leq\left(1-\left(1+\frac{1}{\lambda}\right)\frac{1}{n}\right)^{t}\Phi(\tilde{s})
\end{equation}
By rearranging we get:
\begin{equation}
\E[\Phi(s^t)] \geq \frac{\lambda}{\lambda+1}\left(1-\left(1-\left(1+\frac{1}{\lambda}\right)\frac{1}{n}\right)^{t}\right)\Phi(\tilde{s})\geq \frac{\lambda}{\lambda+1}\left(1-e^{-\frac{\lambda+1}{\lambda}\frac{t}{n}}\right)\Phi(\tilde{s})
\end{equation}
By $(\lambda,\mu)$-closeness we get:
\begin{align*}
\E[SW(s^t)]\geq~& \frac{1}{\mu}\E[\Phi(s^t)]\geq \frac{1}{\mu}\frac{\lambda}{\lambda+1}\left(1-e^{-\frac{\lambda+1}{\lambda}\frac{t}{n}}\right)\Phi(\tilde{s})\\
\geq~&  \frac{1}{\mu}\frac{\lambda}{\lambda+1}\left(1-e^{-\frac{\lambda+1}{\lambda}\frac{t}{n}}\right)\Phi(s^*)\\
\geq~& \frac{\lambda}{\mu}\frac{\lambda}{\lambda+1}\left(1-e^{-\frac{\lambda+1}{\lambda}\frac{t}{n}}\right)SW(s^*)
\end{align*}
By setting the declared bound on $t$ we get the theorem.
\end{proof}

A basic utility game as defined by Vetta \cite{Vetta2002} is a special case of the above theorem, where $\lambda=\mu=1$, i.e.
where the potential is equal to the social welfare. We describe below another class of such games.

A utility congestion game consists of a set of players and a set
of resources $R$. Each players strategy space consists of subsets of the
resources. Each resource $r$ is associated with a utility function $\pi_r$ that
depends only on the number of players $n_r$ using the resource. The utility of
each player is the sum of his utilities from each resource in his strategy. A
utility congestion game is known to admit Rosenthal's potential:
$\Phi(s)=\sum_r \sum_{k=1}^{n_r(s)}\pi_r(k)$. Moreover, the social welfare ends up
being $SW(s)=\sum_r n_r(s)\pi_r(n_r(s))$. It is easy to see that if $\pi_r$ are
decreasing functions then $\Phi(s)\geq SW(s)$. Thus, if we prove submodularity of
the potential and we can bound from above the potential by $\Phi(s)\leq \mu SW(s)$ then
the game will be $(1,\mu)$-smooth.

\begin{theorem}
The potential of any utility congestion game with non-negative decreasing utilities is monotone submodular.
\label{thm:congestion_nas}
\end{theorem}
\begin{proof}
We will view the potential function as a function on multisets, where a player
playing more than one strategies means adding to the resources the extra congestion
that these strategies imply. Thus, we can think of congestion $n_r(s)$ as a set function
that counts how many strategies in the set of strategies $s$ use the resource $r$ (double counting if
a strategy is more than one times in $s$).
Let $s,t$ be two multisets of strategies, Let $s\cup t$ be the multiset sum of
$S$ and $T$. Then:
$$\Phi(s\cup t)=\sum_{r\in R}\sum_{t=1}^{n_r(s\cup t)}\pi_r(t)\geq \sum_{r\in R}\sum_{t=1}^{n_r(s)}\pi_r(t) =\Phi(s)$$
Thus $\Phi(s)$ satisfies monotonicity. Moreover, let $s\subseteq t$ and let $s_i$ be
some strategy of some player. By $s\subseteq t$ we have $\forall r\in R: n_r(s)\leq n_r(t)$.
Also it holds that $\forall r: n_r(s_i\cup s)=n_r(s_i)+n_r(s)$. Then using the decreasing
property of $\pi_r(t)$:
\begin{equation}
 \begin{split}
\Phi(s_i\cup s)-\Phi(s) = & \sum_{r\in R} \sum_{t=n_r(s)+1}^{n_r(s_i\cup s)}\pi_r(t)
=\sum_{r\in R} \sum_{t=n_r(s)+1}^{n_r(s)+n_r(s_i)}\pi_r(t) \\
\geq & \sum_{r\in R}\sum_{t=n_r(t)+1}^{n_r(t)+n_r(s_i)}\pi_r(t)=\Phi(s_i\cup t)-\Phi(t)
 \end{split}
\end{equation}
Thus $\Phi(s)$ satisfies submodularity.
\end{proof}

For instance, when $\pi_r(k) = \frac{v_r}{k}$ then we get that $\Phi(s)\leq H_n SW(s)$. Thus such games are
$(1,H_n)$-smooth and even more interestingly, by applying Theorem \ref{thm:submodular-dynamics} we get that after
$n\cdot \log\left(\frac{1}{2\epsilon}\right)$ rounds of random best-response dynamics the social welfare will be at least $\frac{1}{2\cdot H_n}$ of the optimal.

\begin{algorithm}[tpb]
\SetKwInOut{Input}{Input}\SetKwInOut{Output}{Output}
\BlankLine
\nl Let $s^t$ be the strategy profile at iteration $t$. Initialize $s^0$ to some arbitrary strategy.\\
\nl \For{each iteration $t$}{
\nl Pick a coalitional size $k\in \{1,\ldots,n\}$ inversely proportional to $k$.\\
\nl Pick a coalition $C_t\subseteq [n]$ of size $k$ uniformly at random from all possible coalitions.\\
\nl Let $s_{C_t}^t=\arg\max_{s_{C_t}}\sum_{i\in C_t}u_i(s_{C_t},s_{-C_t}^{t-1})$ be the joint strategy profile of players in $C_t$ that maximizes their total utility, conditional on what the rest of the players are playing.\\
\nl All players in $C_t$ deviate to their strategy in the above optimal. Update $s^{t} = (s_{C_t}^t,s_{-C_t}^{t-1})$.}
\caption{Coalitional Best-Response Dynamics}
\label{alg:best-response}
\end{algorithm}

\section{Ommited Proofs}

\begin{proofof}{Theorem \ref{thm:positive-potential}}
 We will give the proof for the case of utility-maximization games. Consider an arbitrary order of the players and some strategy profile $s$.  By the positive externality condition and the definition of the potential function, we have
\begin{eqnarray*}
u_i(s_{N_i}^*,s_{-N_i}) &\geq& u_i(s_{N_i}^*,s_{-N_i}^{out}) = \Phi(s_{N_i}^*,s_{-N_i}^{out})-\Phi(s_{N_{i+1}}^*,s_{-N_{i+1}}^{out})
\end{eqnarray*}
Combining with the assumption on the relation between potential and social welfare we obtain the coalitional smoothness property:
\begin{align*}
 \sum_i u_i(s_{N_i}^*,s_{-N_i}) \geq~& \sum_i \Phi(s_{N_i}^*,s_{-N_i}^{out})-\Phi(s_{N_{i+1}}^*,s_{-N_{i+1}}^{out})
  =~ \Phi(s^*)-\Phi(s^{out}) \geq \lambda\cdot SW(s^*)
\end{align*}
\end{proofof}

\begin{proofof}{Corollary \ref{cor:any-T}}
Our proof of Theorem \ref{thm:best-response} shows that:
\begin{equation}
SW(s^{t}) \geq \frac{1}{H_n} \left( \lambda \cdot \opt - \mu\cdot SW(s^{t-1})\right)
\end{equation}
Equivalently:
\begin{equation}
SW(s^t) - \frac{\lambda}{H_n+\mu}\opt\geq \frac{\mu}{H_n}\left(\frac{\lambda}{H_n+\mu}\opt - SW(s^{t-1})\right)
\end{equation}
Thus either $SW(s^{t-1})\geq \frac{\lambda}{H_n+\mu}\opt$ or $SW(s^{t})\geq \frac{\lambda}{H_n+\mu}\opt$. Hence, half of the time-steps have such high social welfare, which yields the claimed bound for the empirical distribution.
\end{proofof}

\end{appendix}

\end{document}